\documentclass[reqno,11pt,letterpaper]{amsart}
\usepackage{amsmath,amssymb,amsthm,mathtools,amsfonts,cases}
\usepackage[usenames,dvipsnames]{color}
\usepackage[colorlinks=true,linkcolor=Red,citecolor=Green]{hyperref}

\usepackage[margin=3cm]{geometry}

\def\smallsection#1{\smallskip\noindent\textbf{#1}.}

\let\epsilon=\varepsilon 

\newcommand{\RR}{{\mathbb R}}
\newcommand{\NN}{{\mathbb N}}

\newcommand{\TT}{{\mathbb T}}
\newcommand{\ZZ}{{\mathbb Z}}
\newcommand{\QQ}{{\mathbb Q}}

\newcommand{\OpH}{{H_{\lambda, f, T, x}}}
\newcommand{\Hab}{{H_{[a,b]}(x)}}
\newcommand{\Gab}{{G_{x, E, [a,b]}}}

\newcommand{\floor}[1]{\lfloor #1 \rfloor}

\newtheorem{theorem}{Theorem}
\newtheorem{prop}{Proposition}[section]
\newtheorem{definition}{Definition}

\newtheorem{lemma}{Lemma}[section]
\newtheorem{corollary}[prop]{Corollary}
\newtheorem{remark}{Remark}
\newtheorem{claim}{Claim}

\numberwithin{equation}{section}
\mathtoolsset{showonlyrefs}

\usepackage{enumitem}

\title[Schr\"odinger Operators with Monotone Potentials]{Anderson Localization for Schr\"odinger Operators with Monotone Potentials over Circle Homeomorphisms}

\author{Jiranan Kerdboon}
\address[Jiranan Kerdboon]{University of California, Irvine} 
\email{jkerdboo@uci.edu}

\author{Xiaowen Zhu}
\address[Xiaowen Zhu]{University of Washington} 
\email{xiaowenz@uw.edu}

\begin{document}

\begin{abstract}
    In this paper, we prove pure point spectrum for a large class of Schr\"odinger operators over circle maps with conditions on the rotation number going beyond the Diophantine. More specifically, we develop the scheme to obtain pure point spectrum for Schr\"odinger operators with monotone bi-Lipschitz potentials over orientation-preserving circle homeomorphisms with Diophantine or weakly Liouville rotation number. The localization is uniform when the coupling constant is large enough.
\end{abstract}

\maketitle

\section{Introduction}
The spectral theory of quasiperiodic Schr\"odinger operators has been the subject of extensive study over the past several decades due to its deep origins in physics and the richness of its unusual mathematical features. The general setup of a quasiperiodic operator is given by a family of operators $\OpH$ acting on $\ell^2(\ZZ)$, defined as
\begin{equation}
\label{eq: operatorH}
(\OpH\psi)(n) = \psi(n + 1) + \psi(n - 1) + \lambda f(T^nx)\psi(n),
\end{equation}
where $x\in \TT^1$, $T$ is an irrational rotation on $\TT^1$ defined by $Tx = R_\alpha x = x + \alpha$, with $\alpha\in \RR\setminus \QQ$, and $f:\TT^1\to \RR$ is a potential function. Examples of such operators include $f(x) = \cos(x)$ for the almost Mathieu operator or $f(x) = \tan(x)$ for the Maryland model. One of the most interesting features of quasiperiodic operators is that their spectral type can often be fully characterized by the arithmetic properties of $\alpha$ (and/or $x$) in many situations, as demonstrated in works such as \cite{JL16, HJY21}. Since $R_\alpha$ serves as a fundamental example of general circle homeomorphisms, a natural question arises: If $T$ is not a rotation but a more general circle homeomorphism with rotation number $\alpha$, can we still determine or get some information of the spectral type by the arithmetic properties of $\alpha$?


As one can imagine, the answer may vary depending on properties of $f$, $T$, and $\alpha$. The study of \eqref{eq: operatorH} for general circle diffeomorphisms $T$ was initiated by \cite{JK21} and \cite{K20}. In \cite{JK21}, the authors proved purely continuous spectrum for $f$ that is H\"older continuous, $T$ that is $C^{1 + \operatorname{BV}}$-smooth, and $\alpha$ that is super Liouville. \cite{K20} further explored similar phenomena for circle diffeomorphisms with a critical point or break. On the other hand, for quasiperiodic \eqref{eq: operatorH}, a number of recent papers have proved the opposite, i.e., pure point spectrum, under arithmetic properties of $\alpha$ that go beyond the Diophantine condition, e.g., \cite{L15,JL16,AYZ17,JL18,JLZ20,L20,HJY21,GYZ21,L22}. In this paper, we add to this list by extending the results in \cite{JK15}, where the authors worked with an irrational rotation $T = R_\alpha$ with Diophantine $\alpha$ and potential $f$ satisfying conditions \ref{f1} and \ref{f2} below. We work with the same conditions on $f$ but consider a more general orientation-preserving circle homeomorphism $T$ under the assumptions \ref{T1}:
\begin{enumerate}[label=($\mathcal{F}$\arabic*)]
    \item $f$\label{f1} is one-periodic on $\RR$ and $f(0) = 0$, $f(1 - 0)  := \lim_{x\to1^-} f(x) = 1$.
    \item $f$\label{f2} is bi-Lipschitz monotone, i.e., there exist $\gamma_-, \gamma_+>0$ such that for all $0\leq x<y<1$,
    \[\gamma_-(y - x) \leq f(y) - f(x) \leq \gamma_+(y - x).\]  
\end{enumerate}
\begin{enumerate}[label=($\mathcal{T}$\arabic*)]
    \item \label{T1} Assume the invariant measure of $T$ is denoted by $\nu$ and that 
    \[C_-\nu([x,y])\leq |x - y| \leq C_+ \nu([x,y]).\]
\end{enumerate}

Under these conditions, we obtain results that are similar to the ones in \cite{JK15}. In fact, in addition to extending to more general circle homeomorphisms, we also generalize the result by relaxing the Diophantine condition on $\alpha$ to both weakly Liouville and Diophantine case. Specifically, we prove that
\begin{theorem}[pure point spectrum]
\label{thm: pure_point}
For $f$ satisfying \ref{f1}, \ref{f2} and $T$ satisfying \ref{T1} with weakly Liouville or Diophatine rotation number $\alpha$, or more specifically, $0\leq \beta(\alpha)<\infty$, there is $C_0 = C_0(\gamma_\pm,C_\pm) = O\left(\frac{\gamma_-C_-}{\gamma_+C_+}\right) > 0$ such that for all $\lambda>0$, we have 
\[\sigma_c(\OpH) \cap \left\{E:\beta(\alpha)<C_0L(E;\alpha)\right\} = \emptyset,\quad \forall x\in \TT^1,\]
where $\beta(\alpha)$ and the Lyapunov exponent $L(E;\alpha)$ are defined in Sec \ref{sec: preliminaries}.
\end{theorem}
\begin{remark}
The theorem provides a meaningful statement for homeomorphisms with rotation number $\alpha$ when $\beta(\alpha)$ is small or zero, which corresponds to weakly Liouville or Diophantine $\alpha$ (see Sec \ref{sec: preliminaries}). In fact, the smaller $\beta(\alpha)$ is, the more ``irrational'' $\alpha$ is. For example, since we also proved positivity of Lyapunov exponent $L(E;\alpha)>0$ for all irrational $\alpha$ in Corollary \ref{cor: positive_Lyapunov}, when $\beta(\alpha) = 0$, this implies that $\sigma_c(H_{f,T,x}) = \emptyset$, i.e. the spectrum of $H_{f,T,x}$ is pure point. 
\end{remark}
Note that condition \ref{T1} is equivalent to the existence of a bi-Lipschitz conjugacy between $T$ and $R_\alpha$, meaning that there exists a bi-Lipschitz function $\phi$ that is bounded from above and below such that $\phi\circ T = R_\alpha\circ \phi$. We acknowledge that if such a bi-Lipschitz conjugacy exists and $\alpha$ is Diophantine, our results follow directly from \cite{JK15} by a change of variables: letting $y = \phi(x)$, we obtain $H_{f\circ{\phi^{-1}}, R_\alpha,y}$ and can apply known localization results. However, we choose to present the proof in the more general setting using the invariant measure, since the existence of a bi-Lipschitz (in fact, $C^{1 + \epsilon}$) conjugacy is currently only established for Diophantine $\alpha$.

 
As a byproduct, we also establish Lipschitz continuity of the integrated density of states (IDS, see \eqref{eq: IDS}) for all $\lambda$ in Lemma \ref{lemma: Lipschitz_N}, as well as the continuity and positivity of the Lyapunov exponent for large $\lambda$ in Corollary \ref{cor: positive_Lyapunov}. Together with our key lemma \ref{lemma: exponentiallydecay}, which is uniform in $x, E$, and $\alpha$, these results allow us to achieve uniform localization of $\OpH$ (see Definition \ref{def: uniform_localization}) for sufficiently large $\lambda$ and "sufficiently 
irrational" $\alpha$, i.e. $\beta(\alpha)$ sufficiently small.
\begin{theorem}[Uniform localization]
\label{thm: uniform_localization}
    Let $C_0 = C_0(\gamma_\pm,C_\pm)>0$. If $\lambda>\frac{4 e}{\gamma_-C_-}$ and if $\alpha$ is weakly Liouville with $\beta(\alpha)<C_0\ln\left(\frac{\lambda\gamma_-C_-}{4e}\right)$, then $\OpH$ has uniform localization for all $x$. 
\end{theorem}

We finally mention that a somewhat different proof was developed in \cite{K19} for unbounded lower-Lipschitz monotone $f$ and irrational rotation $T = R_\alpha$ with Diophantine $\alpha$. The key idea is, instead of controlling the change of eigenvalue functions horizontally (see Lemma \ref{lemma: horizontal_distance}), the author controls the change of counting function horizontally. We believe that the Lipschitz continuity of integrated density of states in our proof can be done through that of argument in \cite{K19} and the results there can be generalized to more general $T$ with weakly Liouville $\alpha$ in similar ways to here. This will be explored in a future work.



\smallsection{Structure and key ideas}
Under the assumption of \ref{T1} and allowing weakly Liouville $\alpha$, we re-develop the proof following the method in \cite{JK15} in the key step: We use the non-perturbative proof of localization, first developed in \cite{J99}, together with a detailed analysis of the behavior of box eigenvalues. We provide the latter for general $T$ in Sec \ref{sec: eigenvalue_functions}, which helps with building the large deviation estimates in Sec \ref{sec: LDT}. From the large deviation estimates, we get our key lemma on the uniform exponential decay of generalized eigenfunctions in $x, \alpha, E$ in Sec \ref{sec: localization}. The main results follow immediately in Sec \ref{sec: uniform_localization}. 

The extension of our results from $R_\alpha$ to more general circle homeomorphisms $T$ is based on the observation that the behavior of box eigenvalues is closely related to the distribution of orbits of $T$. While the orbits of $T$ are not evenly distributed with respect to distance, they are evenly distributed with respect to the invariant measure, allowing us to obtain quantitative estimates on their distribution under the comparability assumption \ref{T1}. Appendix \ref{app: orbits} provides the key statements that enable us to carry out this argument. The extension to weakly Liouville $\alpha$, on the other hand, requires a more careful estimate of the decay of generalized eigenfunctions in Sec \ref{sec: localization}.

\section{preliminaries}
\label{sec: preliminaries}
In this section, we will begin by discussing two fundamental concepts: continued fraction expansion and weakly Liouville numbers. Afterward, we will introduce several fundamental properties of discrete Schr\"odinger operators, including the generalized eigenvalue and Schnol's theorem, the Green function and Poisson formula, transfer matrices and Lyapunov exponent, density of states measure, and the Thouless formula.

\smallsection{Notations}
For $x\in \RR$, we use $|x|$ to denote the absolute value and $\Vert x\Vert = \inf\limits_{n\in\ZZ}|x - n|$ to denote the closest  distance between $x\in \RR$ and integers.

\smallsection{Continued fraction expansion and weakly Liouville number}
Any number $\alpha \in [0,1)$ can be written in the continued fraction expansion \cite{V90}:
\begin{align*}
    \alpha = \frac{1}{a_1 +\frac{1}{a_2 + \frac{1}{a_3 + \ldots}}} :=[a_1,a_2,a_3, \ldots].
\end{align*}
with $a_k \in \NN^+$. Let $\frac{p_n}{q_n} = [a_1, \dots, a_n]$ denote the continued fraction approximants. They satisfy 
\begin{equation} \label{eq: recurence_relation}
\begin{split}
p_k & = a_kp_{k-1}+p_{k-2}, ~ p_{-1} = 1, ~ p_0 = 0;\\
q_k & = a_kq_{k-1}+q_{k-2}, ~ q_{-1} = 0, ~ q_0 = 1.
\end{split}
\end{equation}
\begin{definition}[weakly Liouville]
\label{def: weakly_Liouville}
    For $\alpha \in [0,1)$, let
    \[
        \beta(\alpha) = \limsup\limits_{k\to \infty} \frac{\ln q_{k+1}}{q_k}.
    \]
    We call $\alpha$ weakly Liouville if $0<\beta(\alpha)<\infty$. 
\end{definition}
We mention that if $\alpha$ is Diophantine\footnote[1]{$\alpha$ is called \textit{Diophantine} if there is $\kappa>0$ and $\tau>0$ such that $\Vert n\alpha\Vert >\frac{\kappa}{|n|^\tau}$ for all $n$.}, then $\beta(\alpha) = 0$. 
\medskip

For a detailed discussion on the next several definitions, please refer to  \cite[Ch.9,10]{CFKS09} and \cite[Ch.VII]{B68}.

\smallsection{Generalized eigenfunction and Schnol's theorem}
We say $\psi$ is a \textit{generalized eigenfunction} of an operator $H$ with respect to a \textit{generalized eigenvalue} $E$ if $\psi$ is polynomially bounded, i.e. $|\psi(n)|\leq C(1 + |n|)^p$ for some $C>0$, $p\in \NN$ and $H\psi = E\psi$. Schnol's theorem  states that the spectral measure of an operator $H$ is supported on the set of its generalized eigenvalues. 

According to Schnol's theorem, to prove that $H$ has pure point spectrum, it is sufficient to show that all generalized eigenfunctions belong to $\ell^2$. This is because if all generalized eigenfunctions and eigenvalues become eigenfunctions and eigenvalues, respectively, then the spectrum is pure point.

\medskip 

\smallsection{Green function and Poisson formula}
Let $\Hab$ and $\tilde H_{[a,b]}$ denote the restriction of $\OpH$ to $\ell^2([a,b])$ with Dirichlet and periodic boundary conditions, respectively. In particular, for the interval $[a,b] = [0,n-1]$, we use the simplified notations $H_n(x)$ and $\tilde{H}_n(x)$. More specifically, 
\[
\begin{split}
    &H_n(x) =
   \begin{pmatrix}
\lambda f(x) & 1 & & \\
1 & \ddots & \ddots & \\
& \ddots & \ddots & 1 \\
 & & 1 & \lambda f(T^{n-1}x)
\end{pmatrix}_{n\times n,} \\
& \tilde{H}_n(x) =
  \begin{pmatrix}
\lambda f(x) & 1 & & 1\\
1 & \ddots & \ddots & \\
& \ddots & \ddots & 1 \\
1 & & 1 & \lambda f(T^{n-1}x)
\end{pmatrix}_{n\times n.}
\end{split}
\]

Let $\Gab = (\Hab - E )^{-1}$ denote the Green function, and let $\Gab(m,n)$ be its $(m,n)$-entry. Denote $P_n(x,E) =\det(H_n(x)-E)$, and let $P_0(x,E)=1$.

The \textit{Poisson formula} provides a connection between the generalized eigenfunction and the Green function. Specifically, suppose $\psi(n)$ is a generalized eigenfunction of $\OpH$ with respect to generalized eigenvalue $E$, then for $n$ in the interval $[a,b]$, we have the following formula:
\begin{equation}
  \label{Eq: Poisson_formula}
    \psi(n) = -\Gab(a,n) \psi(a - 1) - \Gab(n,b) \psi(b + 1).
\end{equation}

\smallsection{Transfer matrix and Lyapunov exponent}
Rewrite $\OpH \psi = E\psi$ into matrix form:
\[
\begin{pmatrix}
\psi_{n}\\
\psi_{n-1}
\end{pmatrix}
= A_{n-1}(x,E) 
\begin{pmatrix}
\psi_{n-1}\\
\psi_{n-2}
\end{pmatrix}
= A_{n-1}(x,E)\dots A_{0}(x,E)
\begin{pmatrix}
\psi_{0}\\
\psi_{-1}
\end{pmatrix},
\]
where
\[
A_i(x,E) := 
\begin{pmatrix}
E-\lambda f(T^{i}x) & -1\\
1& 0
\end{pmatrix}.
\]
We define the \textit{$n$-step transfer matrix} by \[M_n(x,E) := A_{n-1}(x,E)\dots A_{0}(x,E).\]
One can verify by induction that
\begin{equation}
    \label{eq: M_and_P}
    M_n(x,E) := \begin{pmatrix}
    P_n(x,E) & -P_{n-1}(Tx,E)\\
    P_{n-1}(x,E)& -P_{n-2}(Tx,E)
    \end{pmatrix}.
\end{equation}
The \textit{Lyapunov exponent} is defined to be 
\begin{equation}
    \label{eq: def_of_Lya}
    L(E) := \lim_{n \to \infty}\frac{1}{n}\int_{0}^{1}  \ln \Vert M_n(x,E) \Vert \ d\nu(x).
\end{equation}
\smallsection{Integrated density of states (IDS) and the Thouless formula}
Next, we introduce the density of states measure and the Thouless formula, which connects the Lyapunov exponent of $E$ with the density of states measure. The \textit{integrated density of states} (IDS) is defined as follows:
\begin{equation}
    \label{eq: IDS}
    N(E) := \lim_{n \to \infty}\frac{1}{n}\int_0^1 N_n(x,E) \ d\nu(x), 
\end{equation}
where $N_n(x,E) := \#\sigma(H_n(x)) \cap (-\infty, E]$.
\begin{remark}
We can define $\tilde P_n(x,E)$ and $\tilde N_n(x,E)$, analogous to $P_n(x,E)$ and $N_n(x,E)$ for $H_n$, respectively, for $\tilde H_n(x)$.
\end{remark}
\begin{remark}
Notice that $\tilde H_n(x)$ is a rank-two perturbation of $H_n(x)$. Thus we have $|\tilde N_n(x,E) - N_n(x,E)|\leq 2$.  Thus we can also define the IDS by
\begin{equation}
    \label{eq: IDSbytildeN}
    N(E) = \lim_{n \to \infty}\frac{1}{n}\int_0^1 \tilde N_n(x,E) \ d\nu.
\end{equation}
\end{remark}

The function $N(E)$ is right-continuous, non-decreasing, and approaches zero as $E$ approaches $-\infty$. Its derivative defines a unique probability measure, called the \textit{density of states measure} $N(dE)$. The relation between the density of states measure $N(dE)$ and the Lyapunov exponents $L(E)$ is known as the \textit{Thouless formula}. We state it here without proof, but refer the interested reader to \cite{CFKS09} for more details:
\begin{equation}
\label{Thouless_formula}
    L(E) = \int_\RR \ln |E'-E| \,N(dE).
\end{equation}




\section{Positive Lyapunov Exponent}
\label{sec: eigenvalue_functions}
In this section, we first establish some fundamental properties of box eigenvalue functions, which are the eigenvalues of $\tilde H_n(x)$. We then derive estimates for the distance between these eigenvalue functions. Using these estimates, we obtain the Lipschitz continuity of the IDS $N(E)$ with respect to $E$ and prove the positivity of the Lyapunov exponent $L(E)$ for large $\lambda$.

Recall that $\tilde{H}_n(x)$ is the periodic restriction of $\OpH$ to $[0,n-1]$. Let $\tilde{\mu}_m(x)$, $0\leq m\leq n-1$ be the eigenvalues of $\tilde{H}_n(x)$ in increasing order. We refer to $\tilde{\mu}_m(x)$ as the \textit{box eigenvalue functions}. Now we establish some of their basic properties:
\begin{prop} 
\label{prop: properties_of_eigenvalue_function}
 $\tilde{\mu}_i(x)$ have the following properties:
\begin{enumerate}
    \item $\tilde{\mu}_i(x)$ is $1$-periodic, continuous on $[0,1)$ except at $\{T^{-j}(0)\}_{j = 0}^{n-1}$. By rearranging these discontinuity points in an increasing order, we denote them by $\{\beta_l\}_{l = 0}^{n-1}$. We also denote $I_l := [\beta_l,\beta_{l + 1})$.
    \item \label{prop: Lip_ev} $\tilde{\mu}_i(x)$ is bi-Lipschitz continuous with respect to the invariant measure, and strictly increasing on each $I_l$. In fact, 
    \begin{equation}
        \label{eq: Lip_ev}
        \lambda \gamma_-C_-\nu([x,y]) \leq \mu_i(y) - \mu_i(x)\leq \lambda\gamma_+C_+\nu([x,y]).
    \end{equation}
      \item \label{prop: jump}
    At each jump $\beta_l$, we have 
    \[\tilde{\mu}_i(\beta_l-0) \leq \tilde{\mu}_{i+1}(\beta_l)  \leq \tilde{\mu}_{i+1}(\beta_l-0), \quad 0 \leq i \leq n-2.\]
\end{enumerate}
\end{prop}
\begin{remark}
\label{rmk: Lambda_j}
Because of \eqref{prop: Lip_ev} and \eqref{prop: jump} above, it is natural to define 
\[\Lambda_i(x):= \sum_{l = 0}^{n - 1}\mu_{i + l}(x)\chi_{I_l}(x), \quad 0\leq i \leq n - 1\]
and extend it periodically from $[0,1)$ to $\RR$. As a result, $\Lambda_j(x)$ is monotone increasing on $[\beta_{n - j + 1} + N, \beta_{n - j + 1} + N + 1)$ for any $N\in \ZZ$, and it inherits the properties of $\mu_j(x)$ on each $I_l$. In particular, $\Lambda_j(x)$ is also lower-Lipschitz with respect to invariant measure $\nu$ on $[\beta_{n - j + 1} + N, \beta_{n - j + 1} + N + 1)$:
\[
    \Lambda_i(y) - \Lambda_i(x) \geq \lambda \gamma_-C_-\nu([x,y]), \quad \text{~for~}x<y\text{~and~}x,y\in [\beta_{n - j + 1} + N, \beta_{n - j + 1} + N + 1).
\]
\end{remark}
\begin{proof}\leavevmode
\begin{enumerate}
    \item Note that the box eigenvalue functions $\tilde\mu_i(x)$ are roots of the characteristic polynomial $\tilde P_n(x,E)$. Therefore, each $\tilde\mu_i(x)$ is continuous with respect to the coefficients of $\tilde P_n(x,E)$, which are polynomials of $\{\lambda f(T^jx)\}_{j = 0}^{q_k -1}$. Since $f(T^j x)$ is only discontinuous at $T^{-j}(x)$, $\tilde\mu_i(x)$ is only potentially discontinuous at $\{T^{-j}(x)\}_{j = 0}^{q_k - 1}$.
    \item Notice that for $x<y$ in the same $I_l$, $\tilde H(y) - \tilde H(x)$ is a non-negative diagonal matrix. By Lidskii's theorem (see \cite[Theorem 2]{JKK20}), we have
    \[
        \lambda\min\limits_j  \left(f(T^j y) - f(T^j y)\right)\leq \mu_j(y) - \mu_j(x) \leq \lambda \max\limits_j \left(f(T^j y) - f(T^j x)\right).
    \]
    Notice that  \ref{f2} and \ref{T1} implies that  
    \[f(T^j y) - f(T^j x) \leq \gamma_+C_+\nu([T^jx,T^jy]),\]
    and similarly, $f(T^j y) - f(T_j x)\geq \lambda\gamma_-C_-\nu([x,y])$.
    \item Notice that 
    \begin{equation}
    \label{eq: rank_one}
        \tilde{H}_n(\beta_l-0) -  \tilde{H}_n(\beta_l) = \lambda e_j\otimes e_j
    \end{equation}
   where $0\leq j \leq q_k - 1$ such that $T^{-{(j-1)}}(x) = \beta_l$.
    This leads to the second inequality since $\tilde{H}_n(\beta_l-0) -  \tilde{H}_n(\beta_l)$ is positive semi-definite. To derive the first inequality, by \eqref{eq: rank_one}, let $D$ be the matrix obtained by deleting the row $j$ and column $j$ from  $\tilde{H}_n(\beta_l-0)$ or  $\tilde{H}_n(\beta_l)$. Let $\omega_1 \leq \omega_2 \leq \ldots \leq \omega_{n-1}$ be the eigenvalues of $D$. By eigenvalue interlacing theorem,
    \begin{equation}
    \begin{split}
       \ \ \ \ &\tilde{\mu}_0(\beta_l-0) \leq \omega_1 \leq \tilde{\mu}_1(\beta_l-0) \leq \omega_2 \leq \dots \leq \omega_{n-1} \leq \tilde{\mu}_{n-1}(\beta_l-0),\\
       \ \ \ \ &\tilde{\mu}_0(\beta_l) \leq \omega_1 \leq \tilde{\mu}_1(\beta_l) \leq \omega_2 \leq \dots \leq \omega_{n-1} \leq \tilde{\mu}_{n-1}(\beta_l).
    \end{split}    \end{equation}
   Therefore, $ \tilde{\mu}_{m}(\beta_l-0) \leq \omega_{m+1} \leq \tilde{\mu}_{m+1}(\beta_l)$, for all $0 \leq m \leq n-2$.
\end{enumerate}
\end{proof}

\smallsection{Horizontal comparison}
From now on, we fix $\alpha$, and consider $n=q_k$ since we will use the dynamical properties of the irrational circle map to compare box eigenvalue functions horizontally and vertically. The following lemma provides an upper bound control if we compare the box eigenvalue functions $\tilde{\mu}_i(x)$ and $\tilde{\mu}_i(T^rx)$ horizontally. Note that the estimate is uniform in $r$.
\begin{lemma}
\label{lemma: horizontal_distance}
For any $-q_k + 1\leq r\leq q_k-1$,
\begin{align*}
   \left| {\tilde{\mu}}_i(x) -{\tilde{\mu}}_i(T^{r}x)  \right| \leq \frac{\lambda\gamma_+C_+}{ q_{k+1}}.
\end{align*}
\end{lemma}

\begin{proof}
 Define an $q_k \times q_k$ unitary matrix $S = [{e_{q_k}}, {e_1}, {e_2}, \dots, {e_{{q_k} - 1}}]$
 where $e_j\in \RR^n$ are standard unit vectors. Then
 \[
 \begin{split}
     &S^r \tilde{H}_{q_k}(x)S^{-r}  = \tilde H_{q_k - r}(T^r x) \oplus \tilde H_r(x),\\
     &\tilde H_{q_k}(T^rx) = \tilde H_{q_k - r}(T^r x) \oplus \tilde H_{r}(T^{q_k}x).
 \end{split}
 \]
 By Lemma \ref{lemma: bound_interval}, 
 \[\Vert \tilde H_r(x) - \tilde H_r(T^{q_k} x)\Vert \leq \lambda\max\limits_{0\leq i \leq r - 1} |f(T^r x) - f(T^{q_k + r}x)| \leq \lambda\gamma_+C_+\nu([x,T^{q_k}x]) \leq \frac{\lambda\gamma_+C_+}{q_{k + 1}}.\]
The result follows from Lindskii's theorem.
\end{proof}
\begin{corollary}
\label{cor: true_horizontal}
    For any $x,y\in [0,1)$, 
    \[
        |\tilde\mu_i(x) - \tilde\mu_i(y)|\leq \frac{\lambda\gamma_+C_+}{q_{k + 1}} + \frac{2\lambda\gamma_+C_+}{q_k} \leq \frac{3\lambda\gamma_+C_+}{q_k}.
    \]
\end{corollary}
\begin{proof}
First notice that for given $x\in [0,1)$, depending on which $I_k$ it belongs to, there exists $-q_k + 1\leq\alpha\leq 0$, such that each point in $\{T^rx\}_{r = \alpha}^{\alpha + q_k - 1}$ precisely falls in one interval among $\{I_l\}_{l = 0}^{q_k - 1}$. Thus there is $-q_k + 1\leq r\leq q_k -1$ such that $T^r x $ and $y$ are in the same $I_k$. Then 
\[
    \begin{split}
        |\tilde \mu_i(x) - \tilde \mu_i(y)|&\leq |\tilde \mu_i(x) - \tilde \mu_i(T^rx)| + |\tilde \mu_i(T^rx) - \mu_i(y)|\\
        &\leq \frac{\lambda\gamma_+C_+}{q_{k + 1}} + \lambda\gamma_+C_+\nu([T^rx,y])\\
        &\leq \frac{\lambda\gamma_+C_+}{q_{k + 1}} + \lambda\gamma_+C_+\left(\frac{1}{q_k} + \frac{1}{q_{k + 1}}\right)
    \end{split}
\]
where the first inequality follows from Lemma \ref{lemma: horizontal_distance} and the second follows from Lemma \ref{lemma: bound_interval}.
\end{proof}

\smallsection{Vertical comparison}
We now estimate the lower bound of vertical comparison between eigenvalue functions. Unfortunately, the vertical distance between two closest eigenvalue functions $\tilde\mu_i(x)$ and $\tilde\mu_{i+1}(x)$ is not always positive. However, we can show that at most $M$ eigenvalues can be very close to each other, others will be nicely seperated from them. 
\begin{lemma}
\label{lemma: verticle_distance}
Given $\gamma_\pm$, $\lambda$ and $C_\pm$. For any $\epsilon>0$, there is a $j_0 = j_0(\epsilon)  = \frac{2\gamma_+C_+}{\epsilon\gamma_-C_-}$, such that for any $i,j, q_k$ satisfying $j \geq j_0$ and $0\leq i< i + j \leq q_k - 1$, we have
\begin{equation}
    \label{eq: vertical_dist}
    |\tilde\mu_{i + j}(x) - \tilde\mu_{i}(x)|\geq  \lambda\gamma_-C_-(1 - \epsilon)\frac{j}{q_k}=:d_0(\epsilon)\frac{j}{q_k}.
\end{equation}
\end{lemma}
\begin{proof} 
First notice that given $x$, there exists $-q_k + 1\leq \alpha \leq 0$, such that each point in $\{T^rx\}_{r = \alpha}^{\alpha + q_k - 1}$ falls in precisely one interval among $\{I_l\}_{l = 0}^{q_k - 1}$. Then for any $\alpha\leq r,r'\leq \alpha + q_k - 1$, 
\[
    \begin{split}
    |\tilde\mu_{i + j}(x) - \tilde\mu_i (x)|\geq &   |\tilde\mu_{i + j}(T^{r}x) - \tilde\mu_i(T^{r'}x)| - |\tilde\mu_{i + j}(x) - \tilde\mu_{i + j}(T^{r}x)| - |\tilde\mu_{i}(T^{r'}x) - \tilde\mu_i(x)|\\
    \geq & |\tilde\mu_{i + j}(T^{r}x) - \tilde\mu_i(T^{r'}x)| - \tfrac{2\lambda\gamma_+C_+}{q_{k + 1}}\\
     \geq & \sup\limits_{r,r'}|\tilde\mu_{i + j}(T^{r}x) - \tilde\mu_i(T^{r'}x)| - \tfrac{2\lambda\gamma_+C_+}{q_{k + 1}}.
    \end{split}
\]
In particular, we can pick $r,r'$ such that  $(T^{r'}x, \tilde\mu_i(T^{r'}x))$ and $(T^r x, \tilde\mu_{i + j}(T^{r}x))$ are on the graph of the same $\Lambda_m$, defined in Remark \ref{rmk: Lambda_j}. Put such pairs of $(r,r')$ together and denote the set by $S_j$. Then $[T^rx,T^{r'}x]$ includes $j$ out of $q_k$ subintervals created by the partition $\{T^ix\}_{i = \alpha}^{\alpha + q_k - 1}$ on $[0,1)$, where each intervals have the same invariant measure. Thus by pigeonhole principle,\footnote{In fact, we could bound $|\tilde \mu_{i +j}(T^r x) - \tilde \mu_i(T^{r'}x)|$, the distance between eigenvalue functions, directly by Lemma \ref{lemma: bound_interval} without taking the supremum or referring to the pigeonhole principle. However, the authors choose to prove it this way both because it is more interesting, and because it reveals the uniformity in $x$ in the vertical comparison of eigenvalue functions. It implies that vertical differences of eigenvalue functions at any $x$ is uniformly controlled by the largest vertical differences among all $\Lambda_m$. This observation can be useful in dealing with singular $\nu$ where certain $I_k$'s are too small or the case when $f$ is flat at some $I_k$'s.}
\[
    \sup\limits_{r,r'}|\tilde\mu_{i + j}(T^{r}x) - \tilde\mu_i(T^{r'}x)| \geq \lambda\gamma_-C_- \sup\limits_{r,r'\in S_j} \nu([T^r x, T^{r'}x]) \geq \lambda\gamma_-C_- \frac{j}{q_k}.
\]
Thus
\[
     |\tilde\mu_{i + j}(x) - \tilde\mu_i (x)|\geq \lambda\gamma_-C_- \frac{j}{q_k} - \frac{2\lambda\gamma_+C_+}{q_{k + 1}}\geq \lambda\gamma_-C_-\frac{j}{q_k}\left(1 - \tfrac{2\gamma_+C_+}{\gamma_-C_-}\tfrac{q_k}{q_{k + 1}j}\right)\geq \lambda\gamma_-C_-(1 - \epsilon_0)\frac{j}{q_k}
\]
when $j\geq j_0:= \frac{2\gamma_+C_+}{\epsilon\gamma_-C_-}$.

\end{proof}

\smallsection{Lipschitz continuity of IDS}
Recall that $\tilde{N}_n(x,E) = \#\sigma(\tilde H_n(x)) \cap (-\infty, E]$ and 
\begin{equation}
    \label{eq: def_of_N}
    N(E) = \lim_{n \to \infty}\frac{1}{n}\int_0^1 \tilde N_n(x,E) \ d\nu(x).
\end{equation} 
Now we can derive Lipschitz continuity of $\tilde N_{q_k}(x,E)$ and $N(E)$ from vertical distance of  $\tilde\mu_i(x)$:
\begin{lemma}
\label{lemma: Lipschitz_N}
Given $\lambda, \gamma_+, \gamma_-$, $E$, $E'\in \RR$, we have
\begin{equation}
    \label{eq: Lip_of_N}
    | N(E) - N(E')| \leq \frac{|E - E'|}{\lambda\gamma_-C_-}.
\end{equation}
\end{lemma}

\begin{proof}
Fix $E$ and $E'$. For any $\epsilon>0$, we see from Lemma \ref{lemma: verticle_distance} that any interval of length $d_0(\epsilon)\tfrac{j_0}{q_k}$ contains at most $j_0$ eigenvalues for $q_k$ large enough. This allows us to estimate the number of eigenvalues between $E$ and $E'$:
\[
\begin{split}
     |\tilde N_{q_k}(E,x) - \tilde N_{q_k}(E',x)| &\leq \left(\frac{q_k|E - E'|}{d_0(\epsilon)j_0} + 1\right) j_0\\
     &= \frac{q_k|E-E’|}{d_0(\epsilon)}\left(1 + \frac{j_0d_0(\epsilon)}{q_k|E - E'|}\right)
\end{split}
\]
for any $x$. Let $k\to \infty$, we get 
\[
    |N(E) - N(E')| \leq \liminf\limits_{k \to \infty}\frac{|E - E'|}{d_0(\epsilon)}\left(1 + \frac{j_0d_0(\epsilon)}{q_k|E - E'|}\right) =  \frac{|E - E'|}{d_0(\epsilon)} = \frac{|E - E'|}{\lambda\gamma_-C_-(1 - \epsilon)}.
\]
Since this inequality is true for all $\epsilon$, the result follows. 
\end{proof}

\smallsection{Positivity of Lyapunov exponent}
This is a corollary of Lemma \ref{lemma: Lipschitz_N} which is also useful in the later proof of uniform localization.
\begin{corollary}
\label{cor: positive_Lyapunov}
The Lyapunov exponent $L(E)$ of $\OpH$ is continuous in $E$ and $L(E)$ admits a lower bound
\begin{equation}
    \label{lowerboundofLyapunovexponent}
    L(E) \geq \max \left\{ 0, \ln\left( \frac{\lambda\gamma_-C_-}{2e}\right) \right\}.
\end{equation}
Therefore, $L(E)$ is uniformly positive if $\lambda > \frac{2 e}{\gamma_-C_-}$.
\end{corollary}
\begin{proof}
By Lemma \ref{lemma: Lipschitz_N}, $dN(E)$ is absolutely continuous with respect to $dE$ and the Radon-Nikodym derivative $\frac{dN(E)}{dE}\leq \frac{1}{\lambda\gamma_-C_-}:=\frac{1}{d}$, for $a.e. E$. Thus by Thouless formula,
\[  
    \begin{split}
    L(E) &= \int_\RR \ln|E' - E|dN(E') = \int_\RR (\ln|E' - E|)\frac{dN(E')}{dE'} dE'\geq \int_{E-\frac{d}{2}}^{E + \frac{d}{2}}\frac{1}{d}\cdot \ln|E' - E|\ dE'\\
    & = \frac{2}{d}\int_{0}^{\frac{d}{2}} \ln |E'| \ dE' =\ln \tfrac{d}{2e} = \ln \tfrac{\lambda\gamma_-C_-}{2e},
    \end{split}
\]
where the first inequality follows from monotonicity of $\ln$ function and boundedness of $f(E')$. Finally notice that $L(E)\geq 0$ follows from the definition. Thus we get \eqref{lowerboundofLyapunovexponent}.
\end{proof}



\section{Large deviation theorem}
\label{sec: LDT}
 In this section, we provide two essential ingredients for the proof of localization: Lemma \ref{lemma: small_numerator} provides an upper bound of $P_n(x,E)$ while Theorem \ref{thm: LDT}  provides the large deviation estimate which is central of the non-perturbative proofs of localization, as introduced in \cite{J99}. The first is a result that can be directly adapted from \cite[Lemma 3.5]{JM17}. It holds for arbitrary $\alpha\in \RR\setminus \QQ$ and arbitrary piecewise potentials. 
\begin{lemma}
\label{lemma: small_numerator}
For any $\kappa > 0$ and $E \in \RR$, there exists an $N \in \NN$ such that for all $n>N$ 
\begin{align*}
    |P_n(x,E)| \leq e^{n(L(E)+\kappa)}, \quad \forall x\in [0,1).
\end{align*}
Moreover, $N$ can be chosen to be uniform in $E\in I$ as long as $L(E)$ is continuous on interval $I$.
\end{lemma}
\begin{proof}
This was proved in \cite[Lemma 3.5]{JM17} for irrational rotation $T = R_\alpha$. The same method applies to a general circle diffeomorphism under the assumption \ref{T1}.
\end{proof}



\begin{theorem}[Large deviation theorem]
\label{thm: LDT}
Fix $E$ such that $L(E) > 0$. There exists $C_0 = C_0(\gamma_\pm, C_\pm)>0$
such that for any $0<\delta<L(E)$, there is $k_0$ such that for any $k \geq k_0$, 
\begin{equation}
\label{eq: fake_LDT}
    \nu\{ x \in [0,1) : \tfrac{1}{q_k}\ln|P_{q_k}(x,E)| < L(E)-\delta   \} < e^{-C_0\delta q_k}
\end{equation}
where $m$ is the Lebesgue measure. Moreover, the set on the left-hand side is composed of at most $q_k$ many intervals.
\end{theorem}

\begin{proof}
    Recall that 
    \[
        P_{q_k}(x,E) := \det (H_{q_k}(x)-E) = \prod_{i=0}^{q_k-1} \left(\mu_i(x) -E\right).
    \]
    Denote for convenience
    \[
    f_{q_k}(x):= \frac{1}{q_k}\ln |P_{q_k}(x;E)|  = \frac{1}{q_k}\sum\limits_{i = 0}^{q_k - 1}\ln |\mu_i(x) - E|.
    \]
    Notice that $\mu_i(x)$ is monotone and $f_{q_k}(x) = -\infty$ at $\{x:\mu_i(x) = E\text{~for~some~}i\}$. Thus ``large deviation'' happens near $\{x:\mu_i(x) = E\text{~for~some~}i\}$. The aim is to estimate how large this set can be without rising $f_{q_k}(x)$ too high. The idea is since $\mu_i(x)$ are well-seperated, only the closest (to $E$) several $\mu_i(x)$ contribute the most to the negativity of $f_{q_k}(x)$, the rest are nicely controlled. 
    
    To do so, we split eigenvalues $\tilde \mu_j(x)$ into three clusters: $\mathcal K^+$ above $E$, $\mathcal K^0$ around $E$, and $\mathcal K^-$ below $E$. Notice that by  $|N_{q_k}(x;E) - \tilde N_{q_k}(x;E)|\leq 2$ and Lemma \ref{lemma: verticle_distance}, we can make sure that 
    \begin{enumerate}
        \item The cluster of eigenvalues above $E$, denoted by $\mu_i^+(x)$, $i = 1,2,\cdots$ in an increasing order with $\mu_i^+(x) \geq E + \frac{id_0}{q_k}$. 
        \item The cluster of eigenvalues below $E$, denoted by $\mu_i^-(x)$, $i = 1,2,\cdots$ in an decreasing order with $\mu_i^+(x) \leq E - \frac{id_0}{q_k}$.
        \item The cluster of the rest of eigenvalues, denoted by $\mu_i^0(x)$, with the number of eigenvalues in this cluster does not exceed some $N_0$ uniform in $E$. 
    \end{enumerate}
    For example, this can be achieved by considering the closest $2j_0 + 4$ eigenvalues $\mu_i(x)$ to $E$ to be in the third cluster and every eigenvalue above/below them to be in the first/second cluster. Here $j_0$ is to guarantee the lower and upper bound estimates above and $4 = 2\times 2$ is due to $|N_{q_k}(x;E) - \tilde N_{q_k}(x;E)|\leq 2$. In fact, we can do the same thing for $\tilde \mu_i(x)$, then we just need to pick the closest $2j_0$ eigenvalues instead of $2j_0 + 4$. 
    
    Now decompose $P_{q_k}$, $\tilde P_{q_k}$ correspondingly, 
    \[
        \begin{split}
            &P_{q_k}(x;E) = P_{q_k}^+(x;E)P_{q_k}^0(x;E)P_{q_k}^-(x;E),\\
            &\tilde P_{q_k}(x;E) = \tilde P_{q_k}^+(x;E)\tilde P_{q_k}^0(x;E)\tilde P_{q_k}^-(x;E),
        \end{split}
    \]
    where $P_{q_k}^*(x;E) = \prod\limits_{\mu_i^*\in \mathcal K^*} \mu_i^*(x) - E$, where $*\in \{+,-,0\}$.
\begin{claim}
Let $a,b>0$, 
       \[
        \sum\limits_{j = 1}^n \left[\ln(aj + b) - \ln(aj)\right] \leq \sum_{j = 1}^n \ln \left( 1 + \tfrac{b}{aj}\right) \leq \sum_{j = 1}^n \frac{b}{aj} \leq \frac{b}{a}\ln (n + 1).
    \]
\end{claim}
    By Corollary \ref{cor: true_horizontal} and the claim, we have for any $x,y\in [0,1)$, 
    \[
    \begin{split}
        \left|\ln|\tilde P^{\pm}_{q_k}(x;E)| - \ln|\tilde P_{q_k}^\pm(y;E)|\right| &\leq \sum\limits_{i = 1}^{q_k}\left[\ln \left(\tfrac{jd_0}{q_k} + \tfrac{3\lambda\gamma_+C_+}{q_k}\right) - \ln \left(\tfrac{jd_0}{q_k}\right)\right] \leq C\ln q_k.
    \end{split}
    \] 
    Here we considered all maximum potential perturbation of all $\mu_i(x)$ at the maximum potential place $\{\tfrac{jd_0}{q_k}\}_{j = 1}^{q_k}$. There might be extra terms of $\mu_i^\pm(x)$ that does not pair to $\mu_i^\pm(y)$ but since there are only finitely many terms and  they are bounded, the result is still true with a modification of $C$. For the same reason, the inequality holds for $P_{q_k}(x;E)$ as well. 
    
    Thus there is $L_{q_k}(E)$ such that 
    \begin{equation}
        \label{eq: control_P_pm}
            \begin{split}
                &L_{q_k}(E) \leq \frac{1}{q_k}\ln|P_{q_k}^+(x;E)P_{q_k}^-(x;E)|\leq L_{q_k}(E) + C\frac{\ln q_k}{q_k},\\
                 &L_{q_k}(E) \leq \frac{1}{q_k}\ln|\tilde P_{q_k}^+(x;E)\tilde P_{q_k}^-(x;E)|\leq L_{q_k}(E) + C\frac{\ln q_k}{q_k},
        \end{split}
     \end{equation}
    \begin{claim}
    \label{claim: fake_LDT}
    There is $C_0 = C_0(j_0) = C_0(\gamma_\pm,C_\pm)$ such that for $k$ large enough, for $\delta>0$ small enough,  
    \[
    \nu\{ x \in [0,1) : \tfrac{1}{q_k}\ln|P_{q_k}(x,E)| < L_{q_k}(E)-\delta   \} < e^{-C_0\delta q_k}.
    \]
    \end{claim}
    \begin{proof}
    If $x$ is such that $\frac{1}{q_k}\ln|P_{q_k}(x;E)|\leq L_{q_k} - \delta $, then $\frac{1}{q_k}\ln|P_{q_k}^0(x;E)|\leq -\delta$. Since there are at most $N_0$ eigenvalues in $\mathcal K^0(x)$, thus there is some $l$ such that 
    \begin{equation}
        \label{eq: some_mu}
            \frac{1}{q_k}\ln |\mu_l(x) - E|\leq -\delta/N_0 \quad \Rightarrow |\mu_l(x) - E|\leq e^{-\delta q_k/N_0}.
    \end{equation}
    Among all $x\in [0,1)$, there are at most $q_k$ intervals of $x$ such that some $\mu_l(x) - E$ satisfies  \eqref{eq: some_mu}. In fact, there are at most $q_k$ intersections of  $\overline{\text{graph}(\Lambda_j)}$ and $[0,1)\times \{E\}$. Since each $\Lambda_j$ is monotone, \eqref{eq: some_mu} is only possible for $x$ near such intersections. And by Prop. \ref{prop: properties_of_eigenvalue_function} $\mu_i(x)$ are lower-Lipschitz with respect to invariant measure. Thus for $q_k$ large enough,
    \[
        \nu\{x:\text{~there~is~}l{~s.t.~}|\mu_l(x) - E|\leq e^{-\delta q_k/N_0}\} \leq q_k \frac{e^{-\delta q_k/N_0}}{\lambda\gamma_-C_-}\leq e^{-\frac{\delta q_k}{2N_0}}\leq e^{-C_0\delta q_k}.
    \]
    \end{proof}
    Thus we have proved the result with $L_{q_k}(E)$ instead of $L(E)$. Now we need the last component of the proof: 
    \begin{claim}
    \label{claim: CompareLE}
    For any $\epsilon>0$, $L(E) \leq L_{q_k}(E) + \epsilon$ uniform in $E$ when $q_k$ is large enough.
    \end{claim}
    \begin{proof}
    In fact, since the operator is bounded, $\ln |P_{q_k}^0(x;E)|\leq N_0 C_1$.
    Together with \eqref{eq: control_P_pm}, we get 
    \begin{equation}
        \label{eq: upper_LqkE}
            \frac{1}{q_k}\ln|P_{q_k}(x;E)| \leq L_{q_k}(E) + C\frac{\ln q_k}{q_k} + \frac{N_0C_1}{q_k} \leq L_{q_k}(E) + \epsilon, \forall x\in [0,1)
    \end{equation}
    uniformly in $E$ when $q_k$ is large enough (depending on $\lambda, \gamma_\pm, C_\pm$). The same holds for $\tilde P_{q_k}(x;E)$. 
    
    On the other hand, by Lemma \ref{lemma: small_numerator}, for any $\epsilon>0$, $\frac{1}{n}\ln|P_{n}(x;E)|\leq L(E) + \epsilon$ eventually. While by definition of Lyapunov exponent \eqref{eq: def_of_Lya}, $L(E)$ is the limiting averaging of $\frac{1}{n}\ln\Vert M_n(x;E)\Vert$. But $M_n$ and $P_n$ are connected by \eqref{eq: M_and_P}. Thus we see that on a set of measure at least $1/4$, the following is true for either $n = q_k$, $q_k - 1$ or $q_k - 2$:
    \begin{equation}
        \label{eq: lower_LE}
                \frac{1}{n}\ln|P_n(x;E)|\geq L(E) - \epsilon.
    \end{equation}
    If $n = q_k$, combining \eqref{eq: lower_LE} and \eqref{eq: upper_LqkE} gives us what we want. Otherwise, we first notice by row expansion of determinant, we have
    \begin{align}
        P_n(x;E) + P_{n - 2}(x;E)  = (\lambda f(T^{n - 1}x) - E) P_{n - 1}(x;E) \label{eq: n-1}\\
        \tilde P_n(x;E) + 2(-1)^n = P_n(x;E) - P_{n - 2}(Tx;E).\label{eq: tilde}
    \end{align}
  Then when $n = q_k - 1$, by \eqref{eq: n-1}, we have either $P_{q_k}$ or $P_{q_k - 2}$ satisfies \eqref{eq: lower_LE} so we can combine it with \eqref{eq: upper_LqkE} to derive the result. If $n = q_k - 2$, by \eqref{eq: tilde}, we have either $P_{q_k}$ or $\tilde P_{q_k}$ satisfies \eqref{eq: lower_LE}. For the former case, we get the result. For the latter, combining \eqref{eq: lower_LE} and \eqref{eq: upper_LqkE} with $\tilde P_{q_k}$ instead of $P_{q_k}(x;E)$. The claim follows.
\end{proof}
    Now the result follows immediately: For any $\delta>0$, apply Claim \ref{claim: CompareLE} to get $L(E) \leq L_{q_k}(E) + \delta/2$ eventually so that 
    \[
       \{x\in[0,1): \frac{1}{n}\ln|P_n(x;E)|\leq L(E) - \delta\}\subset\{x\in [0,1): \frac{1}{n}\ln|P_n(x;E)|\leq L_{q_k}(E) - \delta/2\}.
    \]
    Then the result follows from Claim \ref{claim: fake_LDT}.

\end{proof}


\section{Exponential decay of eigenfunctions}
\label{sec: localization}
We prove our key lemma \ref{lemma: exponentiallydecay}, which provides uniform exponential decay of generalized eigenfunction in $x,E,\alpha$. To do so, we introduce some definitions and prove a typical ``either or'' argument in the proof of localization in Lemma \ref{Lemma: Either_or_argument}.
\begin{definition}[Regular point]
    We say a point $n\in \ZZ$ is $(x, c, q_k)$-regular if there is an interval $[a,b]$ with
    \begin{equation}
        \label{Eq: Condition_for_interval_ab}
        n \in [a,b], \ b = a + q_k - 1, \ |a - n| \geq \frac{q_k}{5}, \  |n - b| \geq \frac{q_k}{5},
    \end{equation}
    such that
    \[
            |\Gab(a,n)|\leq e^{-c|n - a|}, \text{~and~} |\Gab(n,b)|\leq e^{-c|n - b|}.
    \]
    Otherwise we say $n$ is $(x,c,q_k)$-singular. 
\end{definition}

\begin{lemma}
    \label{Lemma: Singular_imply_LDT}
    Fix $\delta, E$ such that $0<\delta<L(E)$. For $q_k$ large enough, for any $x$, if $n$ is $(x, L(E) - \delta, q_k)$-singular, then for any $a\in [n - \floor{\frac{3q_k}{4}}, n - \floor{\frac{q_k}{4}}]$,
    \begin{equation}
    \label{eq: LDT}
        |P_{q_k}(T^ax)| \leq e^{q_k(L(E) - \delta/10)}.
    \end{equation}
    Furthermore, let $N_k =  \floor{\frac{3q_k}{4}} - \floor{\frac{q_k}{4}} + 1$ denote the number of such $a$, then \[\frac{q_k + 1}{2} \leq N_k \leq \frac{q_k + 3}{2}.\]
\end{lemma}
\begin{proof}
    Since $n$ is $(x, L(E) - \delta, q_k)$-singular, for any $[a,b]$ satisfying \eqref{Eq: Condition_for_interval_ab}, in particular, for any $a \in [n - \floor{\frac{3q_k}{4}}, n - \floor{\frac{q_k}{4}}]$, $b = a + q_k - 1$, we have  
    \begin{equation}
        \label{Eq: Green_function_too_large}
        \begin{cases}
            \text{either} & |\Gab(a,m)| \geq e^{-(L(E) - \delta)(m - a)},\\
            \text{~or~} & |\Gab(m,b)| \geq e^{-(L(E) - \delta)(b - m)}. 
        \end{cases}
    \end{equation}
Notice that 
    \begin{equation}
        \label{Eq: Expression_of_Green_function}
        \begin{cases}
            |\Gab(a,m)| = \frac{|P_{b - m}(T^{m + 1}x)|}{|P_{q_k}(T^ax)|},\\
            |\Gab(m,b)| = \frac{|P_{m - a}(T^ax)|}{|P_{q_k}(T^ax)|}.
        \end{cases}
    \end{equation}
Now we consider the first case in \eqref{Eq: Green_function_too_large} for simplicity. The other case is similar. By Lemma \ref{lemma: small_numerator}, we have when $q_k$ is large enough
    \begin{equation}
        \label{Eq: Small_numerator}
        |P_{b - m}(T^{m + 1}x)| \leq e^{(L(E) + \delta/10)(b - m)}.
    \end{equation}
    By \eqref{Eq: Green_function_too_large},\eqref{Eq: Expression_of_Green_function} and \eqref{Eq: Small_numerator}, we see that 
    \[
    \begin{split}
        |P_{q_k}(T^ax)| & \leq e^{(L(E) + \frac{\delta}{10})(b - m) + (L(E)  - \delta)(m - a)}\\
        & \leq e^{L(E)(b - a) + \frac{\delta}{10} (b - m) - \delta(m - a)}\\
        & \leq e^{L(E)q_k + 
        \frac{\delta}{10}q_k - \delta \frac{q_k}{5}}\leq e^{(L(E) - \frac{\delta}{10})q_k},
    \end{split}
    \]
    Thus we proved \eqref{eq: LDT}. The bound of $N_k$ follows from direct computation when $q_k \equiv 0,1,2,3 ~(\text{mod~} 4)$. 
\end{proof}
In other words, there are many ``large deviation points'' near each singular point. This fact, together with the large deviation estimates in theorem \ref{thm: LDT} and appropriate weakly Liouville assumption (Definition \ref{def: weakly_Liouville}), leads to the repelling of two singular points. In fact, we prove below that two $(x, L(E) - \delta, q_k)$ singular points are at least ``$q_{k + 1} - q_k/2$'' away from each other:
\begin{lemma}[Either or argument]
    \label{Lemma: Either_or_argument}
    Let $C_0$
    be as in Theorem \ref{thm: LDT}.  Assume $\alpha$ and $E$ satisfy $\beta(\alpha)<C_0 L(E)$. For any $\frac{\beta(\alpha)}{C_0}< \delta<L(E)$, we have that for $q_k$ large enough, and for any $\frac{q_k + 1}{2} < |n - m| \leq q_{k+1} - 1 - \frac{q_k + 1}{2}$, either $m$ or $n$ is $(x, L(E) - \delta, q_k)$-regular for any $x$.
\end{lemma}

\begin{proof}
    WLOG assume $n > m$. For any $\delta<L(E)$, assume both $m$ and $n$ are $(x, L(E) - \delta, q_k)$-singular. By Lemma \ref{Lemma: Singular_imply_LDT}, we have 
    \[
        |P_{q_k}(T^ax)|\leq e^{(L(E) - \delta/10)q_k}
    \] 
    for any $a \in [m - \floor{\frac{3q_k}{4}}, m - \floor{\frac{q_k}{4}}]  \cup  [n - \floor{\frac{3q_k}{4}}, n - \floor{\frac{q_k}{4}}]$. Notice further that  
    \[
        n - \floor{\frac{3q_k}{4}} - (m - \floor{\frac{q_k}{4}}) = n - m - N_k + 1 > \frac{q_k + 1}{2} - \frac{q_k + 3}{2} + 1 = 0.
    \]
    Thus the two intervals of $a$ have no intersection. Overall there are $2N_k \geq q_k + 1$ many possible $a$ such that $|P_{q_k}(T^ax)|\leq e^{(L(E) - \delta/10)q_k}$. By Theorem \ref{thm: LDT} and pigeonhole principle, there are $i, j \in  [m - \floor{\frac{3q_k}{4}}, m - \floor{\frac{q_k}{4}}]  \cup  [n - \floor{\frac{3q_k}{4}}, n - \floor{\frac{q_k}{4}}]$ such that 
    \begin{equation}
    \label{maximum}
        \nu([T^ix,T^jx]) \leq e^{-C_0\delta q_k}.
    \end{equation}
    Notice that $|i - j| \leq n - \floor{\frac{q_k}{4}} - (m - \floor{\frac{3q_k}{4}})  = n - m + N_k - 1 \leq q_{k + 1} - 1.$ By Lemma \ref{lemma: shortest_interval} and \eqref{lemma: bound_interval}, we have 
    \[
    e^{-C_0\delta q_k} \geq \nu([T^ix,T^jx]) \geq \nu([x,T^{q_{k}}x])\geq  \frac{1}{q_{k + 1}}.
    \]
    This implies that 
    \begin{align*}
       C_0\delta <\frac{\ln q_{k + 1} }{q_k} \Rightarrow C_0\delta \Rightarrow C_0\delta\leq \limsup \frac{\ln q_{k + 1}}{q_k} = \beta(\alpha).
    \end{align*}
    which leads to a contradiction with the assumption.
\end{proof}

\begin{lemma}
\label{lemma: exponentiallydecay}
     Let $C_0$ 
    be as in Theorem \ref{thm: LDT}. If  $(x,E,\alpha)$ satisfy
    \begin{enumerate}[label=($\mathcal{E}$\arabic*)]
        \item \label{con: Exp1} $E$ is a generalized eigenvalue of $\OpH$, 
        \item \label{con: Exp2} $\beta(\alpha)<C_0L(E)$,
    \end{enumerate}
    then $E$ is an eigenvalue with exponentially decaying eigenfunction. Denote the normalized  eigenfunction by $\psi$ with $\Vert\psi\Vert_\infty = 1$.
    
    Furthermore, for any $\epsilon>0$, there is a $C = C(\epsilon)$, uniform in all $x,E, \alpha$ which satisfy  \ref{con: Exp1}, \ref{con: Exp2} above,  such that for any pair of eigenvalue $E$ and normalized eigenvector $\psi$, there is $n_0 = n_0(E)$ such that
    \begin{equation}
        \label{eq: uniform_exp_decay}
            |\psi(n)|\leq C(\epsilon)e^{-\frac{1}{10}\left(L(E) - \frac{\beta(\alpha)}{C_0}-\epsilon\right)|n - n_0|}.
    \end{equation}
\end{lemma}
\begin{proof}
    Take any $\frac{\beta(\alpha)}{C_0}<\delta<L(E)$. Let $\psi$ be a generalized eigenfunction of $\OpH$ with respect to $E$. Thus $\psi(n) \leq C_1(1 + |n|)^p$ where $C_1 = C_1(E,x,\alpha)$. We first prove $\psi$ decay exponentially so that $E$ is an eigenvalue, then we prove the decay is uniform in the sense of \eqref{eq: uniform_exp_decay}. 
    
    WLOG assume $\psi(0) \neq 0$. By \eqref{Eq: Poisson_formula}, $0$ is 
    eventually $(x, L(E) - \delta, q_k)$-singular. By Lemma \ref{Lemma: Either_or_argument}, we have for $q_k$ large enough, any $n\in (\frac{q_k + 1}{2}, q_{k+1} - 1 - \frac{q_k + 1}{2}]:= (A_k, B_k]$ is $(L(E) - \delta, q_k)$-regular. Notice further that $A_{k+1} \leq B_{k}$ since $q_{k+1} \geq q_k + 4$ for $k\geq 4$. Thus eventually for any $n$, there is $k$ such that $n \in (A_k,A_{k+1}]$. We derive exponential decay by considering two cases seperately:
    \begin{enumerate}
        \item If $n \in (A_{k}, q_k]$, $n$ is $(L(E) - \delta, q_k)$-regular, by \eqref{Eq: Poisson_formula}, we have for arbitrarily small $\epsilon>0$, eventually
    \begin{equation}
        \label{eq: exp_case_1}
         |\psi(n)|\leq C_1e^{-(L(E) - \delta)q_k/5} (1 + 3n)^p \leq e^{-(L(E) - \delta - \epsilon)n/5}.
    \end{equation}
    \item 
    If $n \in [q_k + 1, A_{k+1}]$, then it is easy to check that $|n - B_k| \geq |n - A_k|\geq n/2$. By \eqref{Eq: Poisson_formula}, we have
    \[
        |\psi(n)| \leq 2e^{-(L(E) - \delta)q_k/5}|\psi(n_1)|
    \]
    where $n_1 = a - 1$ or $b + 1$ for suitable $[a,b]$ satifying \eqref{Eq: Condition_for_interval_ab}. As long as $n_1\in (A_k, B_k]$, where $n_1$ would be $(x, L(E) - \delta, q_k)$ regular, then we can apply \eqref{Eq: Poisson_formula} again to $\psi(n_1)$. We can repeat this process to get $\psi(n_2)$, $\psi(n_3),\dots$, as long as $n_i$ stays in $(A_k,B_k]$. Since $|n - B_k| \geq |n - A_k|\geq n/2$ while $|n_i - n_{i+1}| \leq q_k$, thus we can at least do \[J \geq \frac{|n - A_k|}{q_k}\geq \frac{n}{2q_k}\]
    many times. Then we get
    \begin{equation}
        \label{eq: exp_case_2}
      \begin{split}
         |\psi(n)| &\leq 2^J e^{-(L(E) - \delta)q_kJ/5}|\psi(n_J)|\leq e^{-\left(L(E) -\delta - \frac{5}{q_k}\right)\frac{n}{10}}|\psi(n_J)|\\
         &\leq C_1e^{-(L(E) - \delta - \frac{5}{q_k})n/10}(1 + 3n)^p\leq C_1e^{-(L(E) - \delta - \epsilon)n/10}.
    \end{split}
    \end{equation}
    \end{enumerate}
    Combining \eqref{eq: exp_case_1} and \eqref{eq: exp_case_2} gives us the first half of the theorem. Now since $\psi\in\ell^2$, we can normalize it so that $\Vert\psi\Vert_\infty = 1$. 
    
    The key point of the second half is the uniformity in $x, E, \alpha$. Take $n_0 = \min\{n:\psi(n) = 1\}>-\infty$ to be the leftmost maximum point of $\psi$. By \eqref{Eq: Poisson_formula}, we see that maximum point $n_0$ is always $(x,L(E) - \delta, q_k)$-singular for all $q_k$ . Thus $n$ is $(x,L(E) - \delta, q_k)$-regular if $A_k<|n - n_0|\leq B_k$. We can now repeat the estimates \eqref{eq: exp_case_1} and \eqref{eq: exp_case_2} above with the new, uniform (in $x,E,\alpha$) improvement that $|\psi(n_i)|\leq 1$ instead of $|\psi(n_i)|\leq C_1(x,E,\alpha)(1 + n_i)^p$, where we get 
    \[
        \begin{cases}
            |\psi(n)|\leq e^{-(L(E) - \delta)\frac{|n-n_0|}{5}}, & n\in(A_k,q_k],\\
            |\psi(n)|\leq e^{-(L(E) - \delta - \frac{5}{q_k})\frac{|n - n_0|}{10}}, & n\in (q_k,A_{k + 1}].
        \end{cases}
    \]
 Since $\frac{\beta(\alpha}{C_0}<\delta<L(E)$ is arbitrary and $\frac{5}{q_k}$ is arbitrarily small once $q_k$ is large enough uniformly in $x,E,\alpha$. Thus \eqref{eq: uniform_exp_decay} follows.
\end{proof}
\section{Localization results}
\label{sec: uniform_localization}
Now we prove our main results. Both of them follow directly from Lemma \ref{lemma: exponentiallydecay}:
\begin{proof}[Proof of Theorem \ref{thm: pure_point}]
Recall that by the Schnol's theorem, spectral measure is supported on the set of generalized eigenvalues (see \cite[Ch. VII]{B68}. Fix $\lambda$ and $x$, the theorem follows directly from Lemma \ref{lemma: exponentiallydecay}.
\end{proof}
\begin{definition}[Uniform localization]
\label{def: uniform_localization}
An operator $H$ exhibits uniform localization if 
there exixts $C,c$ such that for any pair of eigenvalue and eigenfunction $E$, $\psi$, there exists $n_0 = n_0(E)$ such that
\begin{align*}
    |\psi (n)| \leq Ce^{-c|n-n_0|}.
\end{align*}
\end{definition}
\begin{proof}[Proof of Theorem \ref{thm: uniform_localization}]
By Corollary \ref{cor: positive_Lyapunov}, $0 < \ln\left(\frac{\lambda\gamma_-C_-}{4 e}\right)\leq L(E)$ for all $E$. It follows that $\beta(\alpha)<C_0 L(E)$. Thus Lemma \ref{lemma: exponentiallydecay} applies to all $x$, all $E$ and those $\alpha$ which satisfy our assumption. By taking $\epsilon = \frac{1}{2}\ln\left(\frac{\lambda\hat{C} \gamma_-}{4\eta e}\right)$ in Lemma \ref{lemma: exponentiallydecay}, we get uniform localization.
\end{proof}

\appendix

\section{Orbital analysis} 
\label{app: orbits}
It is well-known that irrational rotation on 1d-torus, $R_\alpha(x) = x + \alpha$, has best-approximation property, c.f. \cite{V90}, 
\begin{equation}
    \label{eq: irra_best_approximation}
    \Vert q_k\alpha\Vert \leq \Vert n\alpha\Vert, \quad \forall 1\leq n < q_{k + 1}\alpha
\end{equation}
with estimates 
\begin{equation}
    \label{eq: irra_estimate}
    \frac{1}{2q_{k + 1}} \leq \Vert q_k\alpha \Vert \leq \frac{1}{q_{k + 1}}, 
\end{equation}
where $q_k$ is defined in \eqref{eq: recurence_relation}. Furthermore, the orbits of $R_\alpha$ is also well-understood, we cite \cite[Proposition 4.1, 4.2]{JK15} here:
\begin{prop}
    Let $k\geq 1$. The points $\{j\alpha\}$, $j = 0,1,2,\dots,q_k - 1$ splits $[0,1)$ into $q_{k -1}$ ``large'' gaps of length $\Vert (q_k - q_{k - 1})\alpha\Vert$ and $q_k - q_{k - 1}$ ``small'' gaps of length $\Vert q_{k - 1}\alpha\Vert$. Furthermore, we have the estimates 
    \[
    \begin{split}
        &\tfrac{1}{q_k} - \tfrac{q_{k - 1}}{q_kq_{k + 1}}\leq ||q_{k - 1}\alpha|| \leq \tfrac{1}{q_k},\\
        &\tfrac{1}{q_k} \leq ||(q_k - q_{k - 1})\alpha)||\leq \tfrac{1}{q_k} + \tfrac{1}{q_{k + 1}}.
    \end{split}
    \]
\end{prop}

For a general measure-preserving circle homeomorphism $T$ with rotation number $\alpha$, such kind of best approximate properties and orbital analysis can be done in a similar way with invariant measure $\nu$ instead of distance $||\cdot||$. In fact, 
\begin{lemma}[Best approximation]
\label{lemma: shortest_interval}
For any $x\in \TT^1$ and $k \in \NN$, 
\begin{align*}
    \nu([x,T^i x]) \geq \nu([x,T^{q_k}x]) ,
\end{align*}
where $0 \leq i < q_{k+1}$.
\end{lemma}
\begin{proof}
Note that Lemma \ref{lemma: shortest_interval} holds when the invariant measure is the Lebesgue measure - in other words, when the map $T$ is the irrational rotation. 
For a general measure-preserving circle homeomorphism, this inequality holds since it is equivalent to the irrational rotation case. In fact, the Poincar{\'e} classification theorem  \cite[Theorem 4.3.20]{HK03} guarantees the existence of the topological conjucacy $h$ with a rotation $R_\alpha$, and $h$ is also the distribution function for the unique invariant measure $\nu$. Hence, for any $x \in \TT^1$ and  $i \in \NN$, we have
\begin{align*}
    \nu([x,T^ix]) = |h(T^ix)-h(x)| = |R^i_\alpha(h(x))-h(x)| = \| i\alpha \|.
\end{align*}
\end{proof}
\begin{lemma}
\label{lemma: bound_interval}
Fix $x$, Let $k\geq 1$. The points $\{T^j x\}_{j = 0}^{q_k - 1}$ split $[0,1)$ into $q_{k - 1}$ ``large'' gaps of invariant measure $\nu\left([T^{q_k}x,T^{q_{k - 1}}x]\right) = \nu\left([x,T^{q_k - q_{k - 1}}x]\right)$ and $q_k - q_{k - 1}$ ``small'' gaps of invariant measure $\nu\left([x,T^{q_{k - 1}}x]\right)$. Furthermore, we have the estimates 
    \[
    \begin{split}
        &\tfrac{1}{q_k} - \tfrac{q_{k - 1}}{q_kq_{k + 1}}\leq \nu\left([x,T^{q_{k - 1}}x]\right) \leq \tfrac{1}{q_k},\\
        &\tfrac{1}{q_k} \leq \nu\left([x,T^{q_k - q_{k - 1}}x]\right)\leq \tfrac{1}{q_k} + \tfrac{1}{q_{k + 1}}.
    \end{split}
    \]
\end{lemma}
\begin{proof}
 To prove the theorem, let us first introduce the dynamical partition on the circle by following the convention using in \cite{JK21}. For each $k \in \NN$, let $I_k$ be the interval between $x$ and $T^{q_k}x$. It can be verified by induction
 in $k$ that the following collection of intervals forms a \textit{$k^{th}$ dynamical partition} of $\TT^1$
\begin{equation*}
	    \mathcal{P}_k(z) := \{ I_k, T(I_k), \dots, T^{q_{k-1}-1}(I_k)\} \cup \{ I_{k-1}, T(I_{k-1}),  \dots, T^{q_{k}-1}(I_{k-1})\}:= \mathcal S_k \cup \mathcal L_k .
\end{equation*}
That is, they are disjoint except for the endpoints, and the union cover the whole circle. Notice that intervals in $\mathcal S_k$ all have smaller invariant measure $\nu(I_k) <\nu(I_{k- 1})$ than intervals in $\mathcal L_k$, thus we call them ``short'' and ``long'' intervals correspondingly. One can check by induction on $k$ that, each ``long'' interval $T^j(I_{k-1})$ in $k^{th}$ dynamical partition is divided into $a_{k+1}$ ``long'' intervals and one ``short'' interval in $k + 1^{th}$ dynamical partition. More specifically,
\[
T^j(I_{k-1})\in \mathcal L_{k} \Rightarrow \begin{cases}
    T^{j+{q_{k-1}}}(I_k), T^{j+{q_{k-1}}+{q_k}}(I_k), \dots, T^{j+{q_{k-1}}+(a_{k+1}-1){q_k}}(I_k)\in \mathcal L_{k + 1},\\
    T^j(I_{k+1})\in \mathcal S_{k + 1}. 
\end{cases}\]
This allows us to estimates the ``large'' and ``small'' gaps \footnote{Notice that the partition in Lemma \ref{lemma: bound_interval} is different from dynamical partition, ``long'' and ``short'' intervals are also different concepts from ``large'' and ``small'' gaps.} in Lemma \ref{lemma: bound_interval} now.

\begin{proof}[Proof of Theorem \ref{lemma: bound_interval}]
Since $\mu$ is the invariant measure of $T$, for dynamical partition $\mathcal P_{k + 1}(z)$, we have
\begin{equation} 
\label{sum_invariant_measure}
    1 = \sum_{i=0}^{q_{k + 1}-1}\nu( T^i(I_k) )+ \sum_{j=0}^{q_{k}-1} \nu( T^j(I_{k + 1}))=q_{k + 1}\nu( I_k )+q_{k} \nu( I_{k + 1}).
\end{equation}
By \eqref{sum_invariant_measure}, we get 
\begin{align*}
    \nu(I_k) =   \frac{1-q_{k} \nu( I_{k+1})}{q_{k+1}} \leq \frac{1}{q_{k+1}}.
\end{align*}
Moreover, since \eqref{sum_invariant_measure} holds for any $k$, we also get $\nu(I_{k+1})  \leq \frac{1}{q_{k+2}}$. So,
\begin{equation}
    \nu(I_{k}) 
\geq \frac{1}{q_{k+1}}-\frac{q_{k}}{q_{k+1} q_{k+2}} \geq \frac{1}{2q_{k+1}}.
\end{equation}
The last inequality follows from the recurrence relation \eqref{eq: recurence_relation} and $a_k \geq 1$:
\begin{equation}
    q_{k+2} = a_{k+2}q_{k+1} + q_{k} \geq 2q_{k}.
\end{equation}
By the comparability between $\nu$ and the Lebesgue measure on a circle \ref{T1}, the claim follows.
\end{proof}
\end{proof}

\section*{Acknowledgement}
We would like to thank Svetlana Jitomirskaya for suggesting this problem, and helpful discussion and comments; Ilya Kachkovskiy for the helpful discussions on potential sharp conditions and improvements; Sa{\v{s}}a Koci\'c for his comments on the conjugacies. J.K. would also like to thank UCI for their wonderful hospitality. X.Z. was partially supported by Simons 681675, NSF DMS-2052899, DMS-2155211, DMS-2054589. J.K. was partially supported by {the National Science Foundation EPSCoR RII Track-4
Award No. 1738834 and she appreciates Sa\v sa Koci\'c's generous support on the visit to UCI.} 

\bibliographystyle{plain}
\bibliography{ref}

\end{document}